\documentclass{eptcs}
\providecommand{\event}{FROM 2019}
\usepackage{underscore}           
\usepackage{amsmath,amssymb,stmaryrd,txfonts}
\usepackage{tikz}
\usetikzlibrary{automata, positioning, arrows}
\usepackage{pgf}

\tikzset{
->, 
>=stealth, 
node distance=3cm, 
every state/.style={thick, fill=gray!10}, 
initial text=$ $, 
}

\newenvironment{todo}{\medskip\hrule\smallskip\noindent}{\smallskip\hrule\medskip}

\newcommand{\hd}{\mbox{$\mathit{hd}$}}
\newcommand{\eqbydef}{\mbox{$\triangleq$}}
\newcommand{\Paths}{\mbox{$\mathit{Paths}$}}
\newcommand{\fPaths}{\mbox{$\mathit{fPaths}$}}

\newcommand{\hypo}{\mbox{$\mathsf{Hyp}$}}
\newcommand{\conc}{\mbox{$\mathsf{Con}$}}
\newcommand{\lastOcc}{\mbox{$\mathit{last}$}}

\newcommand{\true}{\mbox{\textsc{t}}}
\newcommand{\false}{\mbox{\textsc{f}}}
\newcommand{\post}{\mbox{$\mathit{\partial}$}}

\newcommand{\reach}{\mbox{$\mathit{\,\Rightarrow\!\!\!\Diamond}\,$}}
\newcommand{\rreach}{\mbox{\scriptsize{$\mathit{\,\Rightarrow\!\!\!\Diamond}\,$}}}

\newcommand{\lhs}{\mathit{lhs}}
\newcommand{\rhs}{\mathit{rhs}}
\newcommand{\len}{\mbox{$\mathit{len}$}}
\newcommand{\suf}{\mbox{$\mathit{suf}$}}
\newcommand{\final}{\bullet}
\newcommand{\calP}{\mbox{$\mathcal{P}$}}
\newcommand{\calH}{\mbox{$\mathcal{H}$}}

\newcommand{\calS}{\mathcal{S}}
\newcommand{\calD}{\mbox{$\mathcal{D}$}}

\newtheorem{definition}{Definition}
\newtheorem{lemma}{Lemma}
\newtheorem{theorem}{Theorem}
\newtheorem{example}{Example}
\newenvironment{proof}{\textit{Proof}}{\hfill $\Box$}

\title{(Co)inductive Proof Systems  for Compositional Proofs in Reachability Logic}
\author{Vlad Rusu
\institute{Inria\\ Lille, France}
\and
David Nowak\institute{CRIStAL\thanks{Univ. Lille, CNRS, Centrale Lille, UMR 9189 - CRIStAL - Centre de Recherche en Informatique Signal et Automatique de Lille, F-59000 Lille, France}\\ Lille, France}
}

\def\titlerunning{(Co)inductive Proof Systems  for Compositional Proofs in Reachability Logic}
\def\authorrunning{V Rusu \& D. Nowak}

\begin{document} 

\maketitle


\begin{abstract}
  Reachability Logic is  a formalism that can be used, among others, for expressing partial-correctness properties  of transition  systems.  In
  this paper we present three proof systems for this formalism,  all of which are sound and complete and
  inherit the  coinductive nature of
  the logic.  The proof systems differ, however, in several aspects.  First,
  they use induction and coinduction in different proportions.  
  The second  aspect   regards
  compositionality,   broadly  meaning  their ability  to  prove  simpler
  formulas  on smaller  systems, and  to  reuse those  formulas as lemmas for  more
  complex formulas on larger systems.  The third aspect is the difficulty
  of their soundness proofs. We show that the more induction  a proof system uses,
and the more specialised is its use of coinduction (with respect to our problem domain),
  the more  compositional the proof system is, but the more difficult its soundness proof becomes.
   We also briefly present  mechanisations of these results in
  the Isabelle/HOL and Coq proof assistants.
\end{abstract}

\section{Introduction}
Reachability Logic (RL)~\cite{DBLP:journals/lmcs/StefanescuCMMSR19} has been introduced as a language-parametric program logic: a formalism for specifying the functional  correctness of programs, which may belong to any programming language whose
operational semantics is also specified in RL.  The functional correctness of a program is
stated as the validity of a set of RL formulas (specifying the program's expected properties) with respect to another set of RL formulas (specifying the operational semantics of the language containing the program).

Such  statements are proved by means of a proof system, which has  adequate meta-properties with respect to validity: soundness (i.e., only valid RL formulas can be proved) and relative completeness (all valid RL formulas can, in principle, be proved, modulo the existence of ``oracles'' for auxiliary tasks).  The proof of  meta-properties for the RL proof system is highly nontrivial, but it only needs to be done once.

Program logics already have a half-century history between them,  from the first occurrence of Hoare logic~\cite{DBLP:journals/cacm/Hoare69} to  contemporary separation logics~\cite{DBLP:journals/cacm/OHearn19}. However, all those logics  depend on a
language's syntax and therefore have to be defined over and over again, for each new language (or even, for each new language version).
In particular, the  meta-properties of the corresponding proof systems should  be reproved over and over again, a 
 tedious task that is often  postponed to an indeterminate future.

Despite being language-parametric, Reachability Logic does not come in only one version.
Several  versions of the logic have been proposed over the years~\cite{DBLP:conf/lics/RosuSCM13,DBLP:conf/oopsla/StefanescuPYLR16,DBLP:journals/lmcs/StefanescuCMMSR19}. The formalism has been generalised from programming languages to more abstract models: rewriting logic~\cite{DBLP:conf/birthday/LucanuRAN15,DBLP:conf/lopstr/SkeirikSM17} and transition systems~\cite{DBLP:conf/tase/RusuGH18}, which can be used for specifying designs, and verifying them before they are implemented in program code. This does not replace code verification, just as code verification does not replace the testing of the final running software;  but it enables the early catching of errors and the early discovery of key functional-correctness properties, all of which are known to have practical, cost-effective benefits.

\paragraph*{Contributions.} We further study RL on transition systems (TS). We propose  three proof systems for
RL, and formalise them in the Coq~\cite{DBLP:series/txtcs/BertotC04} and Isabelle/HOL~\cite{DBLP:books/sp/NipkowPW02} proof assistants. One may 
naturally ask: why having several proof systems and  proof assistants - why not one of each? 
The answer is manyfold:
\begin{itemize}
  
\item the proof systems we propose have  some common features: the soundness and completenes meta-properties, and the
  coinductiveness nature inherited from RL. However, they do differ in others aspects:  (i) the ``amount''  of induction they contain; (ii)  their degree of compositionality (i.e., their ability to prove local formulas  on ``components'' of a TS, and then to use those formulas as lemmas in proofs of global formulas on the TS); and (iii) the difficulty level of their soundness proofs.

\item  we show that the more induction  a proof system uses, and the closest its coinduction style to our problem domain of proving reachability-logic formulas, 
  the more  compositional the proof system is, but the more difficult its soundness proof. There is a  winner: the most compositional proof system of the three, but we found that the other ones exhibit  interesting, worth-presenting features as well.

\item Coq and Isabelle/HOL have different styles of coinduction: Knaster-Tarski style vs.\ Curry-Howard style. Experiencing this first-hand with the nontrivial examples constituted by proof systems suggested a spinoff project, which amounts to porting some of the features of one proof assistant into the other one. For the moment, porting  Knaster-Tarski features into the Curry-Howard coinduction of Coq produced promising results, with possible practical impact for a broader class of Coq users.
\end{itemize}

\paragraph{Related Work.}  Regarding RL, most papers in the  above-given list of  references mention its coinductive nature,
but do not actually use it. Several Coq mechanisations of soundness proofs for RL proof systems are presented,  but Coq's coinduction is absent from them.  In~\cite{DBLP:conf/cade/CiobacaL18,DBLP:journals/jsc/LucanuRA17} coinduction is used for formalising RL and  for proving RL properties 
for programs and for term-rewriting systems,
but their approach is not mechanised in a proof assistant.
More closely related work to ours is reported in~\cite{DBLP:conf/esop/MoorePR18}; they attack, however, the problem
exactly in the opposite way: they develop a general theory of coinduction in Coq and use it to verify programs directly based on the semantics of  programming languages, i.e., without using a proof system. They do show that  a proof system for RL is an instance of their approach for theoretical reasons, in order to give a formal meaning to the  completeness of their approach.

Regarding coinduction in Isabelle/HOL, which is based on the Knaster-Tarski fixpoint theorems, we used only a small portion of what is available: coinductive predicates, primitive coinductive datatypes and primitive corecursive functions. More advanced developments are reported in~\cite{DBLP:conf/esop/BlanchetteBL0T17}. Regarding coinduction in Coq, it is based on the Curry-Howard isomorphism that views proofs as programs, hence, coinductive proofs are well-formed corecursive programs~\cite{DBLP:conf/types/Gimenez94}. An approach that bridges the gap between this  and
the Knaster-Tarski style of coinduction is~\cite{DBLP:conf/popl/HurNDV13}.
A presentation of our own  results on porting Knaster-Tarski style coinduction to Coq and a detailed comparison with the above is left for future work.

Regarding coinduction in formal methods, we note that it is mostly used for proving bisimulations. The book~\cite{sangiorgi2011} serves as introduction to both these notions and explores the relationships between them.

Regarding compositional verification, most existing techniques decompose proofs among  parallel composition
operators. Various compositional methods for various parallel composition operators (rely-guarantee for variable-based composition, assumption-commitment for synchronisation-based composition, \ldots) are presented in the book~\cite{DBLP:books/cu/RoeverBH2001}.
We employ compositionality in a different sense - structural, for transition systems, and logical, for formulas.
We note, however, that many of the techniques presented in~\cite{DBLP:books/cu/RoeverBH2001} have a coinductive
nature, which could perhaps be  
 exploited in future versions of RL proof systems.

\paragraph*{Organisation.} The next section recaps preliminary notions: Knaster-Tarski style induction and coinduction, transition systems, and RL on transition systems. A first compositionality result, of RL-validity with respect to certain sub-transition systems, is given. The  three following sections present our three proof systems in increasing order of complexity. Soundness and completeness results are given  and a notion of compositionality  with respect to formulas,  in two versions: asymmetrical and symmetrical, is introduced and combined with the compositionality regarding sub-transition systems. The three proof systems are shown to have increasingly demanding compositionality features. We then briefly discuss the mechanisations of the proof systems in the Coq and Isabelle/HOL proof assistants before we  present future work and conclude.
The Coq and Isabelle/HOL formalisations, as well as a full paper containing proofs of all the results, are available at \url{http://project.inria.fr/from2019}.

\section{Preliminaries}
\label{sec:prelim}
\subsection{Induction and Coinduction}

Consider a complete lattice $(L, \sqsubseteq, \sqcup, \sqcap, \bot, \top)$
and a monotone function $F : L \to L$. According to the Knaster-Tarski fixpoint theorem, $F$ has a least fixpoint $\mu F$ (respectively, greatest fixpoint $\nu F$), which is the least (respectively, greatest) element of $L$ such that $F(x) \sqsubseteq x$ (resp.\ $x \sqsubseteq F(x)$). From this one deduces Tarski's induction and coinduction principles:
$\mbox{$F(x) \sqsubseteq x$ implies $\mu F \sqsubseteq x$}$
and $\mbox{$x \sqsubseteq F(x)$ implies $x \sqsubseteq \nu F$}$.

Those principles can be used
to define inductive and coinductive datatypes and recursive and corecursive functions. For example, the type of natural numbers
is defined as the least fixpoint of the function $F(X) = \{ 0 \} \cup \{ \mathit{Suc}(x) \mid x \in X \}$.
The greatest fixpoint of $F$ is the type of natural numbers with infinity.

As another example, let $\calS = (S, \to)$ be a transition system where $S$ is the set of states and $\to \; \subseteq \; S \times S$ is the transition relation.
A state $s$ is \emph{final}, and we write $\final\,  s$, if there exists no $s'$ such that $s \to s'$.
A path is a nonempty, possibly infinite sequence of states. More formally, the set $\Paths$ of paths   is the greatest fixpoint $\nu F$, where
$F(X) \;\; = \;\; \{ s \; \mid \; \final\, s  \} \; \cup \; \{ s \, \tau \; \mid \; s \in S \; \wedge \; \tau \in X \; \wedge  \; s \to (\hd\, \tau) \}$,  with
$\hd:\Paths \to S$ being simultaneously defined as $\hd(s) = s$ and $\hd(s\, \tau) = s$ for all $s \in S$ and $\tau \in X$.
One can then corecursively define the length of a path as a value in  the  natural numbers with infinity:
$\len\,  s = 0$ and $\len(s\, \tau) = \mathit{Suc}(\len\, \tau)$.

Hereafter, whenever necessary, we emphasise the fact that certain  notions are relative to a  transition system  $\calS$ by postfixing them with $_\calS$. We omit this subscript when it can be inferred from the context.

A complete lattice associated to a transition system $\calS = (S, \to)$,  is the set of state predicates $\Pi$
defined as the set of functions from $S$ to the set of Booleans $\mathbb{B} = \{ \false, \true \}$. Its operations are defined by 
$p \sqsubseteq q \; \eqbydef \; \forall s, p\, s \Rightarrow q\, s$, 
$(p \sqcup q)\, s\; \eqbydef \;\; p\, s \vee q\, s$,  $(p \sqcap q)\, s \; \eqbydef \; p\, s \wedge q\, s$,
$\bot\, s \; \eqbydef \; \false$, $\top\, s \; \eqbydef \; \true$.
We  also extend the transition relation $\to$ of  $\calS$ into a \emph{symbolic transition function} $\post :  \Pi \to \Pi$, defined by
$\post p \; \eqbydef \; \lambda s \;.\; \exists s' \;.\; p\, s' \, \wedge \, s' \to s$.

It is sometimes  convenient to use a stronger variant of Tarski's coinduction principle:
$X \sqsubseteq F(X \sqcup \nu F)$ iff $X \sqsubseteq  \nu F$.
Regarding induction, it is sometimes  convenient to use \emph{continuous} functions, i.e., functions $F$ such that
$F(\bigsqcup_{i \in I} x_i) = \bigsqcup_{i \in I}(F x_i)$, and use Kleene's fixpoint theorem:  $\mu F$ exists and is equal to $\bigsqcup_{n = 0}^{\infty} F^n(\bot)$.

\subsection{Reachability Formulas}
We adapt Reachability Logic  to  transition systems. Assume a transition system $\calS = (S,\to)$.
Syntactically,
a reachability formula (or, simply, a formula) over $\calS$ is a pair $p \reach q$ with  $p,q \in \Pi$. We let $\lhs(p \reach q) \;\eqbydef\; p$ and  $\rhs(p \reach q) \;\eqbydef\; q$. 
We denote by $\Phi_\calS$ the set of all reachability formulas over the transition system $\calS$.

\begin{example}
  Figure~\ref{fig:sum} depicts  an extended finite-state machine having three natural-number variables: $i$,$s$, and $m$, and three control nodes: $c_0$, $c_1$, and $c_2$.
  Arrows connect the nodes and are possibly decorated with a Boolean guard and a set of parallel assignments of the variables. The variable $m$ is never assigned, thus, it stays constant. The  purpose of the machine is to compute in $s$ the sum of the first $m$ natural numbers.

  The machine is a finite representation of an infinite-state transition system whose 
  state-set is the Cartesian product $ \{c_0, c_1, c_2\} \times \mathbb{N}^3$
  and whose  transition relation is
  $\bigcup_{i,s,m \in \mathbb{N}}\{((c_0,i,s,m),(c_0,0,0,m))\}$ $\cup \bigcup_{i, s,m \in \mathbb{N}, i <m}\{((c_1,i,s,m),(c_1,i+1,s+i+1,m))\} \cup \bigcup_{i, s,m \in \mathbb{N}, i \geq m}\{((c_1,i,s,m),(c_2,i,s,m))\}$.
   A  formula expressing the transition systems's functional correctness is
$(c = c_0)\reach (c =c_2 \wedge s =  m\times (m+1)/2)$. 
  
\begin{figure}[t]
\begin{center}
\begin{tikzpicture}
\node[state] (c1) {$c_0$}; 
\node[state, right of=c1] (c2) {$c_1$};
\node[state, right of=c2] (c3) {$c_2$};
\draw
(c1) edge[above] node[yshift = -1.1mm]{\small  $\begin{array}{l} i := 0 \\ s := 0 \end{array}$} (c2)
(c2) edge[loop] node[yshift = 7mm]{\small  $\begin{array}{l} i < m  \\ i := i+1 \\ s := s+i+1 \end{array}$} (c2)
(c2) edge[above] node{\small $i \geq m$} (c3);
\end{tikzpicture}
\end{center}
\vspace*{-4mm}
\caption{\label{fig:sum} Sum up to $m$}
\vspace*{-3mm}
 \end{figure}
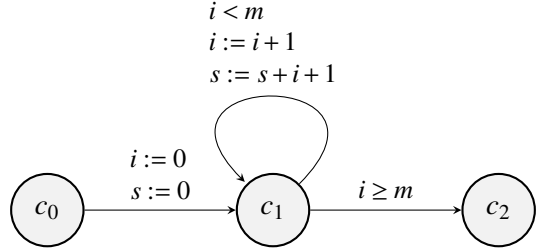
\end{example}  

In order to define the semantics of reachability formulas we first introduce the following relation.

\begin{definition}
  \label{def:leadsto}
  $\leadsto$ is the largest set of pairs $(\tau, r) \in  \Paths \times \Pi$ such that: (i) $\tau = s$ for some $ s\in S$, and $r\, s$;
  or (ii) $\tau = \ s\, \tau'$, for some $s \in S$, $\tau'\in \Paths$, and  $r\, s$;  or (iii)  $\tau = \ s\, \tau'$ for some $s \in S$, $\tau'\in \Paths$,  and $(\tau',r) \in \leadsto$.  
\end{definition}
We write $\tau \leadsto r$
  for  $(\tau,r) \in \leadsto$. Tarski's  principle  induces the following coinduction principle for $\leadsto$:

\begin{lemma} 
\label{lem:coindleadsto}
  For  $R \subseteq Paths \times \Pi$,
 if for all $(\tau, r) \in R$,  it holds that either $ (\exists s. \tau = s \wedge r \, s)$, or  $(\exists s.\exists \tau'.\tau = s\, \tau' \wedge r \, s) $  or $ (\exists s.\exists \tau'. \tau = s\, \tau' \wedge (\tau',r) \in R) $,  
    then $R \subseteq\,  \leadsto$.
\end{lemma}

\begin{definition}[Validity]
  \label{def:valid}
 A formula
$\varphi \in \Phi_{\calS}$ is \emph{valid} over $\calS$, denoted by $\calS \models \varphi$, whenever for all $\tau \in \Paths_{\calS}$ such that $(\lhs\,  \varphi) \, (\hd\, \tau)$ holds, it also holds that $\tau \leadsto_{\calS} (\rhs\, \varphi)$.
\end{definition}

\begin{example}
  \label{ex:valid}
 The formula $(c = c_0)\reach (c =c_2 \wedge s =  m\times (m+1)/2)$ is valid over the transition system denoted by the state-machine depicted in Figure~\ref{fig:sum}. Intuitively, this means that all finite paths ``starting'' in the control node $c_0$ ``eventually reach'' $c_2$ with $s = m\times (m+1)/2$ holding. The ``eventually reach'' expression justifies the $\reach$ notation borrowed from Linear Temporal Logic (LTL). Indeed, reachability formulas are essentially  LTL formulas for a certain version of LTL interpreted over finite paths. 

\end{example}  

We close the section with a simple notion of \emph{component} of  a transition system, and show that, if a formula is valid on a component,
then it is valid on the whole transition system. 

\begin{definition}[Component]
  \label{def:comp}
 A transition system $(S', \to')$ is a \emph{component} of  $(S, \to)$  if 
 \begin{itemize}
  \item $S' \subseteq S$ and $\to' \; \subseteq \; \to$;
 \item  for all $s',s \in S'$,  $s' \to s$ implies $s' \to' s$; 
  \item  for all  $s' \in S'$,  $s \in S\setminus S'$,   $s' \to s$ implies  $s' \in \final_{\calS'}$. 
 \end{itemize}
 We  write   $\calS' \leftslice \calS$ when $\calS'$ is a component of  $\calS$.  
 \end{definition} 
That is,  $\calS'$ is a full sub-transition system of $\calS$,
and  one  may only ``exit''   from  $\calS'$  via its final states.  
We  often interchangeably use sets of states and their characteristic predicates, like we did for $\final_{\calS'}$ above.

\begin{theorem}[Compositionality of $\models$ w.r.t\ transition systems]
  \label{th:tscomp} 
   $\calS' \leftslice \calS$  and $\calS' \models \varphi$ imply $\calS \models \varphi$.
\end{theorem}
\begin{example}
  \label{ex:tscomp}
  In Figure~\ref{fig:sum}, the self-loop on the control node $c_1$
  denotes a transition system $\calS'$ that is a component of the transition system $\calS$ denoted by the whole state machine. Let $\varphi \eqbydef (c = c_1 \wedge i = 0 \wedge s = 0) \reach (c = c_1 \wedge i = m \wedge s = i\times(i+1)/2 )$.
  One can show  that
  $\calS' \models \varphi$, thus, $\varphi$  is also valid over $\calS$.
\end{example}  
One could, in principle, prove the validity of reachability formulas directly from the semantical definitions. However, this has several disadvantages: lack of a methodology - each formula is proved in its own ad-hoc way, and lack of a notion of completeness - is there a uniform way for  proving every valid formula?  These issues are addressed by the  proof systems presented by increasing order of complexity in the next sections.

\section{A One-Rule Proof System}
\label{sec:one-rule}

\begin{figure}[t]
\small
 \begin{align*} 
\textsf{[Stp]~} 
&
\dfrac{\calS \vdash \post l' \reach r}
      {\calS \vdash l \reach r}\nu \ \mathit{\ if \ } l \sqsubseteq l' \sqcup r,  l' \sqcap \final \sqsubseteq \bot
 \end{align*}
 \caption{\label{fig:pf1} One-rule proof system.}
 \vspace*{-3mm}
\end{figure}

Our first proof system is depicted as  the one-rule  inference system in Figure~\ref{fig:pf1}. It is parameterised by a transition system $\calS$, and everything therein depends on it; we omit  $_\calS$ subscripts for simplicity.  Intuitively, an application of the \textsf{[Stp]} rule can be seen as a symbolic execution step, taking a formula $l \reach r$ and ``moving'' $l$ ``one step closer'' to $r$ - specifically, taking an over-approximation $l'$ of the ``difference'' between $l$ and $r$ (encoded in the side-condition $l \sqsubseteq l' \sqcup r$) that contains no final states ($l' \sqcap \final \sqsubseteq \bot$) and performing a symbolic execution step from $l'$ (encoded in the $\partial$ symbolic transition function). The rule is applicable infinitely many times, hence the $\nu$ symbol next to~it.
Note that there are no hypotheses in the proof system: those would be reachability formulas in the left-hand side of the $\vdash$ symbol, not allowed here.

For a more formal definition, consider the function $F : \calP(\Phi) \to \calP(\Phi)$ defined by

$$F(X) =  \bigcup_{l,l'r\in \Pi, \, l \sqsubseteq l' \sqcup r, l\, ' \sqcap \final \sqsubseteq \bot, \, \partial l' \rreach r \in X} \{  l \reach r \}$$   

\noindent  $F$ is monotone, and, by Knaster-Tarski's theorem, $F$ has a greatest fixpoint $\nu F$.
We now define $\calS\vdash \varphi$ 
by $\varphi \in \nu F$.
Tarski's  coinduction principle then  induces the following coinduction principle for~$\vdash$:

\begin{lemma}
  \label{lem:coindvdash1}
For all  set $X\subseteq \Phi$ of hypotheses and $\varphi \in X$, if \emph{for all $l \reach r  \in X $, there is $l' \in \Pi$  such that $l \sqsubseteq l' \sqcup r$, $ l' \sqcap \final \sqsubseteq \bot$ and  $\post l' \reach r\in X$},  then $\calS \vdash \varphi$.
\end{lemma}

\paragraph*{Soundness.} Soundness means that only valid formulas are proved:
\begin{theorem}[Soundness of $\vdash$]
  \label{th:soundness}
  $\calS \vdash \varphi$ implies $\calS \models \varphi$.
\end{theorem}
The proof uses the coinduction principle of the $\leadsto$ relation (Lemma~\ref{lem:coindleadsto}), which occurs in the definition of validity, instantiated with the relation $R \subseteq \Paths \times \Pi$ defined by $R \;\eqbydef\; \lambda (\tau,r). \exists l.( \calS\vdash l \reach r \wedge l \,(\hd \tau))$. As a general observation, all proofs by coinduction use a specific coinduction principle instantiated with a specific predicate/relation. The instantiation step is where the user's creativity is most involved.
  
\paragraph*{Completeness.}
Completeness is the reciprocal to soundness: any valid formula is provable. It is based on the following lemma, which essentially reduces reachability to a form of inductive invariance.

\begin{lemma}
  \label{lem:redtoinv}
If $l \sqsubseteq q \sqcup r$, $q \sqcap \final \sqsubseteq \bot$, and $\post q \sqsubseteq q \sqcup r$ then $\calS \vdash l \reach r$.
\end{lemma}
The proof of this lemma uses Lemma~\ref{lem:coindvdash1} with an appropriate instantion of the set $X$ therein.
\begin{example}
  \label{ex:redtoinv}
  In order to establish $\calS' \models (c = c_1 \wedge i = 0 \wedge s = 0) \reach (c = c_1 \wedge i = m \wedge s = i\times(i+1)/2 )$ - which has been claimed in~Example~\ref{ex:tscomp} - one can use Lemma~\ref{lem:redtoinv} with $q \; \eqbydef \;  (c = c_1 \wedge i < m \wedge s = i \times (i+1)/2) $.
\end{example}  

 \begin{theorem}[Completeness of $\vdash$]
   \label{th:completeness}
   $\calS \models \varphi$ implies $\calS \vdash \varphi$.
 \end{theorem}  
 The  proof of completeness is constructive: it   uses the predicate  $q \;\eqbydef \; \lambda s.\,\neg r s \wedge \forall \tau \in \Paths.(\,s = \hd\, \tau$ $\Rightarrow$ $\tau \leadsto r)$  that, for valid formulas $l \reach r$, is shown to satisfy the three inclusions of Lemma~\ref{lem:redtoinv}.
 One may wonder: even when one does not know whether a  formula $l \reach r$ is valid, can one still use the above-defined $q$ and Lemma~\ref{lem:redtoinv}
 in order to prove it? The answer is negative:  proving the first implication $l \sqsubseteq q \sqcup r$ with the above-defined $q$
 amounts to proving validity directly from the semantics of formulas, thus losing any benefit of having   a proof system.
 Hence, completeness is  a theoretical property; the practically useful property is  Lemma~\ref{lem:redtoinv},  which users have to provide with a suitable $q$ that
 satisfy the three inclusions therein. In~\cite{rusu:hal-01962912} we use this lemmma for verifying  an infinite-state transition-system specification of a  hypervisor.

 \smallskip
 
 Looking back at the  proof system $\vdash$, we note that it is purely coinductive  -  no induction  is present at all.
 This is unlike the proof systems in forthcoming sections. Regarding compositionality (with respect to transition systems)
 our proof system has it, since, by soundness and completeness and Theorem~\ref{th:tscomp}, one has that $\calS' \leftslice \calS$ and $\calS' \vdash \varphi$ implies $\calS\vdash \varphi$.  However, we show below that $\vdash$ does not have another, equally  desirable compositionality feature:~\emph{asymmetrical compositionality} with respect to  formulas.
 
 \paragraph{Asymmetrical compositionality with respect to formulas.}
A proof system with this feature 
 decomposes a proof of a formula $\varphi$ into a proof of a formula $\varphi'$ and one of $\varphi$ assuming~$\varphi'$. The asymmetry between the formulas involed suggested the property's name.
 In Definition~\ref{def:weakcomp} below, $\VDash$ is a binary relation - a subset of $\calP(\Phi) \times \Phi$
(equivalently,  a predicate of type $\calP(\Phi) \to \Phi \to \mathbb{B}$), parameterised by a transition system $\calS$.
For \emph{hypotheses} $\calH \subseteq \Phi$ and $\varphi \in \Phi$, we write $\calS\!, \calH \VDash \phi$ for $(\calH,\phi)  \in \; \VDash$ and $\calS \VDash \phi$ for
$\calS\!, \emptyset \VDash \phi$.

\begin{definition}[Asymmetrical compositionality with respect to formulas]
  \label{def:weakcomp}
  A  proof system $\VDash$ is asymmetrically compositional with respect to formulas  if
  $\calS \!\VDash\varphi'$ and $\calS\!, \{\varphi'\} \VDash \varphi $ imply $\calS\! \VDash \varphi$.
\end{definition}
The proof system $\vdash$ is not asymmetrically compositional w.r.t.\ formulas, because that requires hypotheses,
which $\vdash$  does not have. One could add hypotheses to it, and a new rule to prove a formula if it is found among the hypotheses.
However,  note that, unlike the \textsf{[Stp]} rule, the new  rule has an inductive nature: it can only  occur a finite number of times in a $\vdash$ proof (specifically, at most once, at the end of a finite proof).

\section{An Asymmetrically-Compositional Proof System}

In this section we propose another proof system $\Vdash$ and show that it is  compositional with respect  to transition systems and asymmetrically compositional with respect to  formulas. These gains are achieved
thanks to a the introduction of inductive rules in the proof system, enabling a better distribution of roles between these rules and the remaining coinductive rule;  all at the cost of a more involved soundness proof.

\begin{figure}[t]
\small
\begin{align*}
 \textsf{[Hyp]~} 
&
\dfrac{}
      {\calS\!,\calH \Vdash \varphi} \mu\ \mathit{\ if \ } \varphi \in \calH\\
\textsf{[Trv]~} 
&
\dfrac{}
      {\calS\!,\calH \Vdash r \reach r}\mu\\
\textsf{[Str]~} 
&
\dfrac{\calS\!,\calH \Vdash l' \reach r}
      {\calS\!,\calH \Vdash l \reach r}\mu \ \mathit{\ if \ } l \sqsubseteq l'\\
\textsf{[Spl]~}       
&
\dfrac{\calS\!,\calH \Vdash l_1 \reach r \qquad \calS\!,\calH \Vdash l_2 \reach r}
      {\calS\!,\calH \Vdash (l_1 \sqcup l_2) \reach r} \mu\\
\textsf{[Tra]~}       
&
\dfrac{\calS\!\models  l \reach m \qquad \calS\!,\calH \Vdash m \reach r}
      {\calS\!,\calH \Vdash l \reach r} \mu \\                
\textsf{[Stp]~} 
&
\dfrac{\calS\!,\calH  \Vdash \post l \reach r}
      {\calS\!,\calH \Vdash l \reach r}\nu \mathit{\ if \ }  l \sqcap \final \sqsubseteq \bot
 \end{align*}
   \caption{\label{fig:pf2} Mixed inductive-coinductive proof system.} 
\end{figure}
Our second proof system is depicted in Figure~\ref{fig:pf2}. It is a binary relation - a subset of
$\calP(\Phi) \times \Phi$ (or equivalently, a binary predicate of type $\calP(\Phi) \to \Phi \to \mathbb{B}$),  parameterised by a transition system $\calS$. 
Intuitively, the rule  \textsf{[Stp]}, labelled with $\nu$, is coinductive, i.e., it can be applied infinitely many times, and
the rules \textsf{[Hyp]}, \textsf{[Trv]}, \textsf{[Str]}, \textsf{[Spl]}, and \textsf{[Tra]}, labelled by $\mu$ are inductive, i.e., they can only be applied finitely many times between two consecutive applications of \textsf{[Stp]}. Stated differently, a proof in $\Vdash$ is a possibly infinite series of \emph{phases}, and in each phase there are finitely many applications of \textsf{[Hyp]}, \textsf{[Trv]}, \textsf{[Str]}, \textsf{[Spl]}, and \textsf{[Tra]} and,  except in the last phase (if such a last phase exists), one application of \textsf{[Stp]}.

Note that making the inductive rules coinductive would compromise soundness, because, e.g., the
\textsf{[Str]} rule could forever reduce a proof of any formula to itself, thus proving any formula, valid or not.

The roles of the rules are the following ones.  \textsf{[Hyp]} allows one to prove a formula if it is among the hypotheses.   \textsf{[Trv]}  is in charge of proving trivially valid formulas. \textsf{[Str]} is a general principle that amounts to strengthening a formula before proving it. \textsf{[Spl]} is used for getting rid of disjunctions in left-hand sides of formulas,
which occur when several, alternative symbolic behaviours are explored in a proof search.
 \textsf{[Tra]} is a transitivity rule, used for proving facts about sequential symbolic behaviour.
 Note also the asymmetry in hypotheses of the rule \textsf{[Tra]}: for one formula validity is required, while for the other one, it is provability. This asymmetry is used to avoid technical difficulties that arise when proving the soundness of~$\Vdash$, but, as we shall see, it generates difficulties of its own.
Finally, \textsf{[Stp]} makes the  connection between the concrete paths  and the symbolic ones, which the proof system explores during proof search.

For a formal definition: consider the following functions from $\calP(\Phi)$ to $\calP(\Phi)$ defined by

\begin{itemize}
 \item $\vdash_{{\cal S,H},\, Y}^{\mathsf{[Hyp]}}(X) = \calH$
  
\item $\vdash_{{\cal S,H},\, Y}^{\mathsf{[Trv]}}(X) = \bigcup_{r \in \Pi} \{r \reach r\}$

\item  $\vdash_{{\cal S, H},\, Y}^{\mathsf{[Str]}}(X) = \bigcup_{l,l',r \in \Pi, \, l' \rreach r \in X, \,  l \sqsubseteq l'} \{l \reach r\}$
 \item $\vdash_{{\cal S, H},\, Y}^{\mathsf{[Spl]}}(X) = \bigcup_{l_1, l_2, r\, \in \Pi,\,  \{l_1 \rreach r,\, l_2 \rreach r\} \subseteq X} \{(l_1 \sqcup l_2) \reach r\}$
 \item  $\vdash_{{\cal S, H},\, Y}^{\mathsf{[Tra]}}(X) = \bigcup_{l,r,m \in \Pi, \, \calS \models l \rreach m, \; m\rreach r \in X} \{l \reach r\}$

\item $\vdash_{{\cal S, H},\, Y}^{\mathsf{[Stp]}}(X) =  \bigcup_{l,r\in \Pi, \, l \sqcap \final \sqsubseteq \bot, \, \partial l \rreach r \in Y} \{  l \reach r \}$    
\end{itemize}  
Let  $\vdash_{{\cal S, H}, Y}\!(X) \!= \vdash_{{\cal S, H}, Y}^{\mathsf{[Hyp]}}\!(X) \cup\!\vdash_{{\cal S, H}, Y}^{\mathsf{[Trv]}}\!(X)  \cup\!\vdash_{{\cal S, H}, Y}^{\mathsf{[Str]}}\!(X)   \cup\!\vdash_{{\cal S, H}, Y}^{\mathsf{[Spl]}}\!(X)   \cup\!\vdash_{{\cal S, H}, Y}^{\mathsf{[Tra]}}\!(X)  \cup\!\vdash_{{\cal S, H}, Y}^{\mathsf{[Stp]}}\!(X)$.
 It is not hard to show that $\vdash_{{\cal S,H},\, Y} : \calP(\Phi) \to \calP(\Phi)$ is continuous, thus, by the Knaster Tarski and Kleene fixpoint theorems
 it has a smallest fixpoint $\mu\!\vdash_{{\cal S, H},\, Y} =\bigcup_{n=1}^{\infty} \vdash_{{\cal S, H},\, Y}^n(\emptyset)$.
Now,  we define the function $F_{{\cal S, H}} : \calP(\Phi) \to \calP(\Phi)$ by
$F_{{\cal S, H}}(Y) = \mu\!\vdash_{{\cal S, H}, \, Y}$.   $F_{{\cal S, H}}$ is monotone, thus,  it has a greatest fixpoint
$\nu F_{{\cal S, H}} = \nu (\lambda Y.\mu\!\vdash_{{\cal S, H}, \, Y})$ $ = \nu\mu\!\vdash_{\cal S, H}$.

We  define the proof system $\Vdash$ as follows : for all $\calH  \subseteq \Phi$ and $\varphi \in \Phi$, 
$\calS\!,\calH \Vdash\phi $ iff  $\varphi \in \nu\mu\!\vdash_{\cal S, H}$.
The  inductive-coiductive nature of  $\Vdash$ is visible from its definition. It admits the following coinduction principle:

\begin{lemma}
  \label{def:coindps2}
If  $ X \subseteq\mu\!\vdash_{{\cal S, H}, \ X}$
then for all
$ \varphi \in X$ it holds that  $\calS\!, \calH \Vdash\varphi$.
\end{lemma}

\paragraph*{Using the coinduction principle.} For proving statements of the form $\calS\!, \calH \Vdash\varphi$, one can:
\begin{itemize}

 \item   find a sequence  $X = X_0, \cdots X_n = \emptyset$ of sets
   such that $X_i \; \subseteq  \; \vdash_{{\cal S, H}\, , X} (X_{i+1})$, for $i = 0, \ldots, n-1$, and $\varphi\in X$;
\item  since  
  $\mu\!\vdash_{{\cal S, H}, \, X } =\bigcup_{n=1}^{\infty} \vdash_{{\cal S, H}, \, X}^n(\emptyset)$, we obtain by induction on $n$
that   $X_i \subseteq \mu\!\vdash_{{\cal S, H}\, ,X}$ for  $i = 0, \ldots, n-1$ and in particular
  $X \subseteq  \mu\!\vdash_{{\cal S, H}\, , X }$. By Lemma~\ref{def:coindps2}, $\calS\!, \calH \Vdash\varphi$.
\end{itemize}  

We illustrate the  above approach by proving a key lemma for the completeness of $\Vdash$.

  \begin{lemma}
  \label{lem:redtoinv2}
If $l \sqsubseteq q \sqcup r$, $q \sqcap \final \sqsubseteq \bot$, and $\post q \sqsubseteq q \sqcup r$ then $\calS\Vdash l \reach r$.
\end{lemma}
  \begin{proof}
    We apply the above  approach. Note that $\calH = \emptyset$.
   We choose  $X= X_0 = \{l \reach r, q \reach r, \post q \reach r\}$.
\begin{itemize}
\item Let
  $X_1  = \{(q \sqcup r)\reach r,  q \reach r, \post q \reach r\}$;  using the hypothesis $l \sqsubseteq q \sqcup r$,
  $X_0  \subseteq \vdash^{\mathsf{[Str]}}_{{\calS}\, , \emptyset, \, X}(X_1) \subseteq \vdash_{{\calS}\,, \emptyset, \, X}(X_1)$;

\item Let $X_2 =\{q \reach r, r\reach r, \post q \reach r\}$ ; we obtain
  $X_1  \subseteq  \vdash^{\mathsf{[Spl]}}_{{\calS},\emptyset,\, X} (X_2) \subseteq \vdash_{{\calS},\emptyset,\, X}(X_2)$;
    \item Let $X_3 = \{q \reach r\,, \post q \reach r\}$; we obtain $X_2 \subseteq \vdash^{\mathsf{[Trv]}}_{{\calS},\emptyset,\, X}(X_3) \subseteq \vdash_{{\calS},\emptyset,\, X}(X_3)$;
    \item Let $X_4 = \{\post q \reach r\}$; using  the second  hypothesis $q \sqcap \final \sqsubseteq \bot$ and  the fact that $\partial q \reach l \in  X $ we obtain
      $X_3 \subseteq \vdash^{\mathsf{[Stp]}}_{{\calS},\emptyset,\, X}(X_4) \subseteq\vdash_{{\calS},\emptyset,\, X} (X_4)$
      ;
    \item Let  $X_5 = \{(q \sqcup r) \reach r\}$ ; using the  hypothesis $\post q \sqsubseteq q \sqcup r$, we obtain
  $X_4 \subseteq \vdash^{\mathsf{[Str]}}_{{\calS}, \emptyset,\, X}(X_5) \subseteq \vdash_{{\calS}, \emptyset,\, X}(X_5)$;
    \item Let $X_6 = \{q \reach r, r\reach r\}$; we obtain $X_5 \subseteq \vdash^{\mathsf{[Spl]}}_{{\calS},\emptyset,\, X}(X_6) \subseteq \vdash_{{\calS},\emptyset,\, X}(X_6)$;

\item    Let $X_7 = \{q \reach r\}$; we obtain $X_6 \subseteq \vdash^{\mathsf{[Trv]}}_{{\calS},\emptyset,\, X}(X_7) \subseteq \vdash_{{\calS},\emptyset,\, X}(X_7)$;

\item  Let $X_8 = \emptyset$; using the second hypothesis $q \sqcap \final \sqsubseteq \bot$ and  the fact that $\partial q \reach l \in  X$, we obtain $X_7 \subseteq \vdash^{\mathsf{[Stp]}}_{{\calS},\emptyset,\, X}(X_8) \subseteq \vdash_{{\calS},\emptyset,\, X}(X_8)$.
\end{itemize}
Hence, by basic properties of inclusion, $X\subseteq \bigcup_{n = 0}^{7} \vdash^{n}_{{\calS},\emptyset,\, X} (\emptyset) \subseteq \bigcup_{n = 0}^{\infty} \vdash^{n}_{{\calS},\emptyset,\, X} (\emptyset) = \mu\!\vdash _{{\calS},\emptyset,\, X}$, and from $l \reach r \in X$ and Lemma~\ref{def:coindps2}  we obtain $\calS\vdash l \reach r$. 
\end{proof}

  \paragraph*{Soundness.}
  We define the recursive function $\suf :  \{\tau \in \Paths \} \to \{i : \mathbb{N} \mid i \leq (\len\,  \tau)\}  \to \Paths$ by 
  $\suf\,  \tau \, 0 = \tau$ and $ \suf (s\, \tau) (i + 1) = \suf \, \tau \, i$. Intuitively, $\suf \, \tau \, i$ is the sequence obtained by removing $i \leq (\len \, \tau)$ elements from the ``beginning'' of $\tau$.
  This  is required in the definition of the following relation and is used hereafter.

\begin{definition}
  \label{def:hookright}
  $\hookrightarrow \subseteq \Paths \times \Pi$ is the largest set of pairs $(\tau, r)$ such that: (i) $\tau = s$ for some $ \in S$ such that $r\, s$;
  or (ii) $\tau = \ s\, \tau'$, for some $s \in S$, $\tau'\in \Paths$ such that  $r\, s$;  or (iii)  $\tau = \ s\, \tau'$ for some $s \in S$, $\tau'\in \Paths$  and $n \leq (\len\,  \tau')$
  such that $((\suf\, \tau' \, n),r) \in \hookrightarrow$.
\end{definition}
We write $\tau \hookrightarrow r$  instead of $(\tau,r) \in \hookrightarrow$. By analogy with Lemma~\ref{lem:coindleadsto} (coinduction principle for the  $\leadsto$ relation), but using Tarski's strong induction principle, we obtain: 

\begin{lemma} 
\label{lem:coindhookright}
 Let $R\!\subseteq\!Paths\!\times\!\Pi$ be s.t.\  $(\tau, r) \in R\!\Rightarrow\!(\exists s. \tau\!=\!s\!\wedge\!r \, s)\!\vee\!(\exists s.\exists \tau'.\tau = s\, \tau' \wedge r \, s) \vee (\exists s.\exists \tau' .\exists n.\exists \tau''. \tau = s\, \tau'
  \wedge \tau'' =  (\suf\,  \tau'\, n) \wedge (\tau'',r) \in R \vee \tau''\hookrightarrow r)) $.  
     Then $R \subseteq\,  \hookrightarrow$.
\end{lemma}
The following lemma is easily proved, by instantiating the parameter relation $R$, which occurs
in both the coinduction principles  of the relations $\leadsto, \hookrightarrow$, with the other relation:

\begin{lemma}[$\leadsto$ equals  $\hookrightarrow$]
  \label{lem:eqrel} 
  For all $\tau \in \Paths$ and $r \in \Pi$, $\tau \leadsto r$ if and only if $\tau \hookrightarrow r$.
\end{lemma}
Using the coinduction principle for $\hookrightarrow$ and the above equality, as well as the induction principle for the functional $\vdash_{ {\cal S, H}, \, \nu\mu\vdash_{\cal S, H}  } $
we obtain, in a rather involved proof  mixing induction and coinduction:

\begin{theorem}[Soundness of $\Vdash$]
  \label{th:soundps2}
  If for all $\varphi' \in \calH$,  $\calS \models \varphi'$, then $\calS\!,\calH \Vdash \varphi$ implies $\calS\models \varphi$.
\end{theorem}

   \begin{example}
   \label{ex:ps2}
   We sketch a proof of the fact that the transition system $\calS$ denoted by the state machine in Figure~\ref{fig:sum}  meets its functional correcteness property:
  (i) $\calS \models (c = c_0)\reach (c =c_2 \wedge s =  m\times (m+1)/2)$. We first show 
   (ii) $\calS \models (c = c_0) \reach (c = c_1 \wedge i = 0 \wedge s = 0)$, which can be done using
   in sequence the rules \textsf{[Stp]}, \textsf{[Str]}, and \textsf{[Trv]} of the $\Vdash$ proof system together with its soundness.
   Using (ii) and the \textsf{[Tra]} rule, (i) reduces to
    proving (iii) $\calS \Vdash  (c = c_1 \wedge i = 0 \wedge s = 0)\reach (c =c_2 \wedge s =  m\times (m+1)/2)$. Next, in Examples~\ref{ex:tscomp} and~\ref{ex:redtoinv}
   we established\footnote{Example~\ref{ex:redtoinv} used the proof system $\Vdash$ and its Lemma~\ref{lem:redtoinv}, but $\Vdash$ and its corresponding Lemma~\ref{lem:redtoinv2} can be used just as well.} $\calS \models (c = c_1 \wedge i = 0 \wedge s = 0) \reach (c = c_1 \wedge i = m \wedge s = i\times(i+1)/2 )$, hence, using this and the \textsf{[Tra]} rule,  (iii)  reduces to proving
   $\calS \Vdash (c = c_1 \wedge i = m \wedge s = i\times(i+1)/2 ) \reach (c = c_2 \wedge s = m\times(m+1)/2 )$. This  is performed by applying in sequence the rules \textsf{[Stp]}, \textsf{[Str]}, and \textsf{[Trv]}, which concludes the proof.
 \end{example}      

 \paragraph*{Completeness.}
 By  analogy with Theorem~\ref{th:completeness}, but using Lemma~\ref{lem:redtoinv2} instead of Lemma~\ref{lem:redtoinv}:
 
  \begin{theorem}[Completeness of $\Vdash$]
    \label{th:completeness2}
   $\calS \models \varphi$ implies $\calS\Vdash \varphi$.
  \end{theorem}

  \paragraph*{Compositionality.}
  Remembering Definition~\ref{def:weakcomp} of asymmetrical compositionality w.r.t\ formulas:
  
  \begin{theorem}
    \label{th:ps2comp} $\Vdash$ is  asymetrically compositional with respect to formulas.
  \end{theorem}
  \begin{proof}
    We have to show that if  (i) $\calS\Vdash \varphi'$ and (ii) $\calS\!,  \{\varphi'\} \Vdash \varphi $ then $\calS\Vdash \varphi$.   
    Now, (i) and (ii) and the soundness of $\Vdash$ imply  $\calS \models \varphi'$ and  then $\calS \models \varphi$, and then
    the conclusion $\calS\Vdash \varphi$ holds by the completeness of $\Vdash$.
  \end{proof}
  
     Note that the statement (i) can be replaced by a weaker $\calS'\!\Vdash \varphi'$  for  components $\calS' \leftslice \calS$, thanks to the 
     soundness and completenesss of $\Vdash$ and of Theorem~\ref{th:tscomp}. This allows us to mix 
     the compositionality of $\Vdash $ with respect to transition systems and the asymetrical one with respect to formulas.

     The $\Vdash$ proof system is thus better at compositionality than $\vdash$, thanks to the inclusion of inductive rules, in particular, of the rule $\textsf{[Hyp]}$, but
     at the cost of a more involved soundness proof.
  It still has a problem: the asymmetry of the $\textsf{[Tra]}$ rule, required by the soundness proof, is not  elegant since the rule
  mixes semantics $\models$ and syntax $\Vdash$. This is not only an issue of elegance, but a practical issue as well.
  \begin{example}
    \label{ex:astra}
    We attempt to prove the property (i) from Example~\ref{ex:ps2} using the asymmetrical  compositionality
    of $\Vdash$ w.r.t formulas. The first step, similar to that of Example~\ref{ex:ps2},
    is proving  (ii') $\calS \Vdash (c = c_0) \reach (c = c_1 \wedge i = 0 \wedge s = 0)$ by using
    in sequence the rules \textsf{[Stp]}, \textsf{[Str]}, and \textsf{[Trv]} of  $\Vdash$.
    Then, Theorem~\ref{th:ps2comp} reduces (i) to (iii') $\calS,\{(c = c_0) \reach (c = c_1 \wedge i = 0 \wedge s = 0)\} \Vdash  (c = c_0)\reach (c =c_2 \wedge s =  m\times (m+1)/2)$.
    The natural next step would be to use  the  \textsf{[Tra]} rule of  $\Vdash$,
    splitting  (iii') in two parts:  $\calS,\{(c = c_0) \reach (c = c_1 \wedge i = 0 \wedge s = 0)\} \Vdash  (c = c_0)\reach (c = c_1 \wedge i = 0 \wedge s = 0)$,  discharged by \textsf{[Hyp]}, and
   then $\calS,\{(c = c_0) \reach (c = c_1 \wedge i = 0 \wedge s = 0)\} \Vdash  (c = c_1  \wedge i = 0 \wedge s = 0)\reach  \reach (c = c_1 \wedge i = 0 \wedge s = 0)$. But the \textsf{[Tra]} rule of  $\Vdash$, as it is, does not allow this. Hence, when one uses compositionality, one may get stuck in proofs because of technical issues with \textsf{[Tra]}.
  \end{example}
These issues are solved in the third proof system, which incorporates even more induction that the second one, and specialises its coinduction even closer to our problem domain. The third proof system also has better compositionality features. These gains come at the cost of an even more involved soundness proof.

\section{A Symmetrically-Compositional Proof System}
 
\begin{figure}[t]
\small
\begin{align*}
 \textsf{[Hyp]~} 
&
\dfrac{}
      {\calS\!,\calH \Vvdash (\true,\varphi)} \mu\ \mathit{\ if \ } (\false,\varphi) \in \calH\\
\textsf{[Trv]~} 
&
\dfrac{}
      {\calS\!,\calH \Vvdash (b,r \reach r)}\mu\\
\textsf{[Str]~} 
&
\dfrac{\calS\!,\calH \Vvdash (b,l' \reach r)}
      {\calS\!,\calH \Vvdash (b,l \reach r)}\mu \ \mathit{\ if \ } l \sqsubseteq l'\\
\textsf{[Spl]~}       
&
\dfrac{\calS\!,\calH \Vvdash (b,l_1 \reach r) \qquad \calS\!,\calH \Vvdash (b,l_1 \reach r)}
      {\calS\!,\calH \Vvdash (b,(l_1 \sqcup l_2)) \reach r} \mu\\
\textsf{[Tra]~}       
&
\dfrac{\calS\!, \calH  \Vvdash  (b,l \reach m) \qquad \calS,\calH \Vvdash (b,m \reach r)}
      {\calS\!,\calH \Vvdash (b,l \reach r)} \mu \\                
\textsf{[Stp]~} 
&
\dfrac{\calS\!,\calH  \Vvdash (\true,\post l \reach r)}
      {\calS\!,\calH \Vvdash (b,l \reach r)}\mu \mathit{\ if \ }  l \sqcap \final \sqsubseteq \bot \\
 \textsf{[Cut]~} 
&
\dfrac{\calS\!,\calH  \Vvdash (\false, \varphi') \qquad  \calS\!,\calH\cup\{(\false,\varphi')\}  \Vvdash (b, \varphi)}
      {\calS\!,\calH \Vvdash (b, \varphi)}\mu \\
 \textsf{[Cof]~} 
&
\dfrac{\calS\!,\calH\cup\{(\false,\varphi)\}  \Vvdash (\false, \varphi)}
      {\calS\!,\calH \Vvdash (b, \varphi)}\mu\\
  \textsf{[Clr]~} 
&
\dfrac{\calS\!,\calH \Vvdash (b, \varphi)}
      {\calS\!,\calH\cup\{(b',\varphi')\}  \Vvdash (b, \varphi)}\mu     
 \end{align*}
   \caption{\label{fig:pf3}  Inductive proof system, with coinduction managed in hypotheses.} 
\end{figure}

Our third proof system is depicted in Figure~\ref{fig:pf3}. A first difference with the previous one is that hypotheses and conclusions
are pairs of a Boolean tag and a formula. We call them tagged formulas, or simply formulas when there is no risk of confusion. The role of the tags is  to avoid unsoundness.

  The second difference is that the proof system is essentially inductive, i.e., there are no more infinite proofs, and
 no coinduction principle any more; whatever coinduction remains is tailored to our problem and emulated by the proof system, as can be seen seen below in the description of the proof system's rules.

Another difference, especially with the second proof system $\Vdash$, is that the hypotheses set is not constant.  The following rules change the hypotheses set.  First, the \textsf{[Cut]} rule, which says that in order to prove $(b,\varphi)$ under hypotheses $\calH$, it is enough to prove $(\false,\varphi')$
  -  for some formula $\varphi'$ - under hypotheses $\calH$, and to prove  $(b,\varphi)$ under $\calH \cup \{(\false,\varphi')\}$. This resembles a standard cut rule, but it is taylored to our specific setting. Second, the \textsf{[Cof]} rule adds a ``copy'' of the conclusion in the
  hypotheses, but tagged with \false\ -  and the new conclusion is also tagged with $\false$. It is called  this way in reference to the Coq \texttt{cofix} tactic that builds coinductive proofs in Coq  also by copying a conclusion in the hypotheses; hence, we emulate in our proof system's hypotheses a certain existing coinduction mechanism, and taylor it to proving reachability formulas. Note that, without the tags,  one could simply assume any formula by \textsf{[Cof]} and prove it by \textsf{[Hyp]}, which would be unsound since it would prove any
  formula, valid or not. Third, the \textsf{[Clr]} rule  removes a formula from the hypotheses.
  Note that the \textsf{[Stp]} rule, when applied bottom to top, switches the Boolean from whatever value $b$ it has to $\true$. Hence, it is  \textsf{[Stp]} that makes ``progress'' in our setting, enabling the use of \textsf{[Hyp]} in a sound way.
  The other rules have the same respective roles as their homonyms in  $\Vdash$.

\paragraph*{Soundness.} We present the 
  soundness proof of $\Vvdash$ at a higher level of abstraction than for
  the other proof systems. For example, we define  $\Vvdash$-proofs as finite trees, and   assume that finite trees are known to the readers.
  For the other proof systems we adopted a more formal approach
  because the  proofs in those systems were certain kinds of possibly infinite trees, whose a priori knowledge
  cannot be  assumed.
  
  \begin{definition}[Proof]
    A \emph{proof} of a tagged formula $(b,\varphi)$ for a transition system $\calS$  and under hypotheses $\calH$ - for short,
    a proof of $\calS\!,\calH \Vvdash (b,\varphi)$  
    - is a finite tree, whose root is labelled by the sequent  $\calS\!,\calH \Vvdash (b,\varphi)$, whose nodes are also labelled by sequents, obtained by applying bottom-up the rules depicted in Figure~\ref{fig:pf3}.
  \end{definition}  
  We sometimes just write $\calS\!,\calH \Vvdash (b,\varphi)$ for ``there is a proof of  $\calS\!,\calH \Vvdash (b,\varphi)$'' as defined above. The following definition introduces the sets of all hypotheses and of all conclusions occuring in a proof.
  
  \begin{definition}[All hypotheses and conclusions occuring  in proof] Assume a proof $\Theta$  of $\calS\!,\calH \Vvdash (b,\varphi)$. The set $\hypo$ is the union of all sets $\calH'$ of  formulas, for all the node-labels $\calS\!,\calH' \Vvdash (b',\varphi')$
    in  the tree $\Theta$. The set $\conc$  is the set of all  formulas  $(b',\varphi')$, for all the node-labels $\calS\!,\calH' \Vvdash (b',\varphi')$ occuring in~$\Theta$.
  \end{definition}
Hereafter in the currenct subsection about soundness we assume a proof (tree) $\Theta$ of $\calS\!,\calH \Vvdash (b,\varphi)$ with corresponding sets $\hypo$ and $\conc$.
The following technical lemma is  proved by structural induction on such trees. It says that tagged fomulas in  $\hypo$ are  among the
hypotheses  $\calH$ present at the root of  $\Theta$, plus the conclusions $\conc$, and, except perhaps
for those in $\calH$, the formulas in $\hypo$ are tagged with
  $\false$.

  \begin{lemma}
    \label{lem:aux}
    $\hypo \subseteq  \calH \cup \conc$, and,  if $(b', \varphi') \in \hypo \setminus \calH$, then $b' = \false$.
  \end{lemma}

  Some  more notions need to be defined.
First, a \emph{pad} in a tree  is a sequence of consecutive edges, and the length of a pad is the number of nodes on the pad. Hence, the length of a pad is strictly positive.
  \begin{definition}
    The \emph{last occurence of a tagged formula $(b',\varphi') \in \conc$ in $\Theta$ } is the maximal length of a pad from the root $\calS\!,\calH \Vvdash (b,\varphi)$ of $\Theta$ to some node labelled by $\calS, \calH' \Vvdash (b',\varphi')$. For formulas  $(b',\varphi') \notin \conc$ we define by convention their last occurence in $\Theta$ to be $0$. This defines a total function $\lastOcc : \mathbb{B} \times  \Phi \to \mathbb{N}$. 
  \end{definition}
  Let also $\fPaths$ denote the set of finite paths of the transition system under consideration.
  We now define the set $\calD \;\eqbydef\; \{(\tau',b',\varphi') \in  \fPaths \times \mathbb{B} \times \Phi \mid (\lhs\, \varphi')(\hd\, \tau') \wedge (b',\varphi') \in \conc\}$ on which we shall reason by well-founded induction.
  We equip $\calD$ with a well-founded order, namely, with the restriction to $\calD$ of the lexicographic-product order on $\fPaths \times \mathbb{B} \times \Phi$  defined by $(\tau_1,b_1,\varphi_1) \prec (\tau_2,b_2,\varphi_2) $ iff

  \begin{enumerate}
 \item   $\len\, \tau_1 < \len,\tau_2$, or
 \item $\len\, \tau_1 = \len,\tau_2$ and $b_1 < b_2$, with $<$ on Booleans is defined by $\false < \true$, or
 \item  $\len\, \tau_1 = \len,\tau_2$ and $b_1 = b_2$, and $\lastOcc(b_1,\varphi_1) > \lastOcc(b_2,\varphi_2)$.
   \end{enumerate}
  The first two orders in the product, on natural numbers and on Booleans, are well-founded. For the third one, since the order $\prec$ on $\fPaths \times \mathbb{B} \times \Phi$ is restricted to $\calD$,
  all  last occurences are bounded by the height of $\Theta$, ensuring that the inequality  $\lastOcc(b_1,\varphi_1) > \lastOcc(b_2,\varphi_2)$ induces a well-founded order.  Hence, the restriction of $\prec$  on $\calD$ (also denoted by $\prec$) is a well-founded order as well. The following lemma uses this.

 \begin{lemma}
   \label{lem:key}
   Assume $\calS\!,\calH \Vvdash (\false, l \reach r)$ and  for all $(b',\varphi')\in \calH$, $b' = \false$ and $\calS\models \varphi'$.
   Let $\calD$ be the domain corresponding to $\calS\!,\calH \Vvdash (\false, l \reach r)$.
   Then, for all $(\tau,b,\varphi) \in \calD$,  there is
   $k \leq \len\, \tau$ such that $(\rhs\, \varphi)\, (\tau\,  k)$.
 \end{lemma}    
 
 As a corollary to Lemma~\ref{lem:key} we obtain:

 \begin{theorem}[Soundness of $\Vvdash$]
    \label{th:soundp32}
  If for all $(b',\varphi') \in \calH$, $b' = \false$ and $\calS \models \varphi'$, then $\calS\!,\calH \Vvdash (\false,\varphi)$ implies $\calS\models \varphi$.
 \end{theorem}

 \paragraph*{Completeness.} Proving the completeness of $\Vvdash$ is the same as for the other proof system: prove a lemma reducing reachability to an invariance property and then show that for valid formulas that property holds.

 \begin{lemma}
     \label{lem:redtoinv3}
If $l \sqsubseteq q \sqcup r$, $q \sqcap \final \sqsubseteq \bot$, and $\post q \sqsubseteq q \sqcup r$ then $\calS\Vvdash  (\false, l \reach r)$.
 \end{lemma}  
 \begin{proof}
   We build a proof (tree) for $\calS\Vvdash (\false, l \reach r)$. The root of the  tree is a node $N_0$ labelled $\calS\Vvdash  (\false, l\reach r)$.
   $N_0$  has one successor $N_1$, generated by the \textsf{[Str]} rule, thanks to the hypothesis $l \sqsubseteq q \sqcup r$, and labelled
   $\calS\Vvdash  (\false, (q \sqcup r)\reach r)$. $N_1$ has two successors $N_{2,1}$ and $N_{2,2}$, generated by the \textsf{[Spl]} rule, and labelled
   $\calS\Vvdash  (\false, q \reach r)$ and $\calS\Vvdash  (\false, r \reach r)$, respectively.
Usinng the \textsf{[Trv]} rule, $N_{2,2}$  has no succesors.
    $N_{2,1}$ has one successor $N_3$, generated by the \textsf{[Cof]} rule, labelled
   $\calS\!, \{   (\false, q \reach r) \} \Vvdash  (\false, q \reach r)$. $N_3$ has one successor
   $N_4$, generated by the
   \textsf{[Stp]} rule, thanks to the hypothesis  $q \sqcap \final \sqsubseteq \bot$, and labelled
   $\calS\!, \{   (\false, q \reach r) \} \Vvdash  (\true, \post q \reach r)$. Note that the Boolean has switched from
   \false\ to \true, which enables us to later use  the \textsf{[Hyp]} rule. The node $N_4$ has one successor, 
   generated by the \textsf{[Str]} rule thanks to the hypothesis $\post q \sqsubseteq q \sqcup r$:
   $\calS\!, \{   (\false, q \reach r) \} \Vvdash  (\true,  (q \sqcup r) \reach r)$. $N_4$ has two successors $N_{5,1}$ and  $N_{5,2}$, labelled 
    $\calS\!, \{   (\false, q \reach r) \} \Vvdash  (\true, q \reach r)$ and  $\calS\!, \{   (\false, q \reach r) \} \Vvdash  (\true, r \reach r)$, respectively. Neither has any successor: $N_{5,1}$, by the \textsf{[Hyp]} rule, and $N_{5,2}$, by the $\textsf{[Trv]}$ rule. 
 \end{proof}

 By  analogy with Theorems~\ref{th:completeness} and \ref{th:completeness2}  but using Lemma~\ref{lem:redtoinv3} (instead of 
\ref{lem:redtoinv} and~\ref{lem:redtoinv2}, respectively) :

\begin{theorem}[Completeness of $\Vvdash$]
    \label{th:completeness3}
   $\calS \models \varphi$ implies $\calS\Vvdash \varphi$.
  \end{theorem}

  \paragraph*{Compositionality w.r.t. Formulas} $\Vvdash$ has a symmetrical version of   compositionality w.r.t.\ formulas:
  \begin{theorem}
    \label{lem:compform}
    $\calS\!, \calH \cup \{ (\false,\varphi_1)\} \Vvdash (\false, \varphi_2)$ and  $\calS\!, \calH \cup \{ (\false,\varphi_2)\} \Vvdash (\false,\varphi_1)$
    imply  $\calS\!, \calH  \Vvdash (\false, \varphi_1)$ and  $\calS\!, \calH  \Vvdash (\false,\varphi_2)$.
  \end{theorem}
  \begin{proof}
    The statement is symmetrical in $\varphi_1, \varphi_2$; we prove it for the first formula. The rule\textsf{[Cof]}
    generates one successor for the root $N_0$ labelled $\calS\!, \calH  \Vvdash (\false, \varphi_1)$:
  $N_1$, labelled   $\calS\!, \calH \cup \{ (\false,\varphi_1)\} \Vvdash (\false, \varphi_1)$. From $N_1$, the rule \textsf{[Cut]}
    generates two successors, $N_{2,1}$ labelled $\calS\!, \calH \cup \{ (\false,\varphi_1)\} \Vvdash (\false, \varphi_2)$, which we assumed as a hypothesis,
    and $N_{2,2}$, labelled $\calS\!, \calH \cup \{ (\false,\varphi_1), (\false,\varphi_2)\} \Vvdash (\false, \varphi_1)$. From $N_{2,2}$ the rule \textsf{[Clr]}
    removes the first hypothesis and generates a node labelled $\calS\!, \calH \cup \{ (\false,\varphi_2)\} \Vvdash (\false, \varphi_1)$, which we assumed as a hypothesis as well.
  \end{proof}

  \begin{example}
    \label{ex:asymcompps3}
    In Example~\ref{ex:astra} we tried to prove $\calS \models (c = c_0)\reach (c =c_2 \wedge s =  m\times (m+1)/2)$ using the asymmetrical compositionality of $\Vdash$, and noted that a certain
    proof step was impossible because of the asymmetry of the \textsf{[Tra]} rule of $\Vdash$.
    We show that $\Vvdash$ does not suffer from the same issue. The problem, reformulated in terms of $\Vvdash$, was to start the sequent
  (iii')  $\calS,\{(\false,(c = c_0) \reach (c = c_1 \wedge i = 0 \wedge s = 0))\} \Vdash  (\false, (c = c_0)\reach (c =c_2 \wedge s =  m\times (m+1)/2))$ and  to use  the  \textsf{[Tra]} rule
  in order to  split this sequent in two:  $\calS,\{(\false,(c = c_0) \reach (c = c_1 \wedge i = 0 \wedge s = 0))\} \Vdash (\false, (c = c_0)\reach (c = c_1 \wedge i = 0 \wedge s = 0))$ and
   then $\calS,\{(\false,(c = c_0) \reach (c = c_1 \wedge i = 0 \wedge s = 0))\} \Vdash (\false, (c = c_1  \wedge i = 0 \wedge s = 0)\reach  \reach (c = c_1 \wedge i = 0 \wedge s = 0))$. This inference step, which we have just performed above, was not a problem for the $\Vvdash$ proof system.
 \end{example}   
  
Finally, we show how to combine compositionality w.r.t. transition systems and w.r.t. formulas.
  The following lemma says that $\Vvdash$ is compositional w.r.t.
  transition systems even in the presence of hypotheses. 
    
  \begin{lemma}
    \label{lem:compts}
    If $ \calS'\!, \calH  \Vvdash (b,\varphi)$ and  $\calS' \leftslice \calS$ then
    $ \calS\!, \calH \Vvdash (b,\varphi)$.
  \end{lemma}

  Combining Theorem~\ref{lem:compform} and Lemma~\ref{lem:compts} we obtain as a corollary the following theorem, which combines symmetrical compositionality w.r.t. formulas and compositionality w.r.t. transition systems.

  \begin{theorem}
    \label{th:compps3}
    If, for $i  \in \{0, 1\}$, $\calS_i \leftslice \calS$  and
    $\calS_i, \calH \cup\{(\false, \varphi_{1-i})\} \Vvdash (\false,\varphi_i)$,  then,  for $i  \in \{0, 1\}$, $\calS\!, \calH \Vvdash (\false,\varphi_i)$.
  \end{theorem}

\begin{example}
    \label{ex:gcd}
    We sketch the verification of another infinite-state transition system, denoted by the
    state machine in Figure~\ref{fig:gcd}, which computes the greatest common divisor of two strictly positive natural numbers. The obtained proof is not, by far, the simplest; for such simple systems a global (non-compositional) proof is much shorter. Our goal here is to use all the  compositionality features of $\Vvdash$ embodied in Theorem~\ref{th:compps3}. 
    
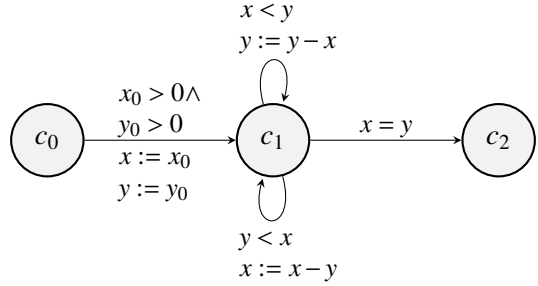
\begin{figure}[t]
\begin{center}
  \begin{tikzpicture}
\node[state] (c0) {$c_0$};     
\node[state, right of=c0] (c1) {$c_1$}; 
\node[state,right  of=c1] (c2) {$c_2$};

\path[->] (c0) edge [above] node[yshift= -10mm]{ \small  $\begin{array}{l} x_0> 0 \wedge \\ y_0 > 0   \\ x := x_0 \\ y := y_0 \end{array}$} (c1);
\path[->] (c1) edge [loop above] node[yshift=-1mm,  xshift=2mm]{\small  $\begin{array}{l} x < y \\ y := y - x \end{array}$} ();
\path[->] (c1) edge  [loop below] node[yshift=1mm, xshift=2mm] {\small  $\begin{array}{l} y < x \\ x := x - y \end{array}$} ();
\path[->] (c1) edge [above] node[yshift=-1mm]{\small  $x = y$} (c2);
 \end{tikzpicture}
\end{center}
\caption{\label{fig:gcd} Computing a greatest common divisor}
 \end{figure}

The state machine has three control nodes and operates with four natural-number variables: $x$, $y$, $x_0$ and $y_0$. The last two variables are ``symbolic constants'', not modified by the transitions of the state machine, whose greatest-common divisor the machine is supposed to compute. On the leftmost transition $x$ and $y$ are initialised to $x_0$ and $y_0$, provided that
the guard $x_0 > 0 \wedge y_0> 0$ holds. On the upper self-loop arrow, $x$ is substracted from $y$ provided the guard $x < y$ holds. The lower self-loop arrow inverses the roles of $x$ and $y$. The rightmost arrow is taken provided its guard
$x = y$ holds.  The state-machine denotes an infinite-state transition system $\calS$ with state-set
$\{c_0, c_1, c_2\} \times \mathbb{N}^{4}$  and transition relation 
{\small
$\bigcup_{x,y,x_0,y_0 \in \mathbb{N}, x_0 > 0, y_0 > 0}\{((c_0,x,y,x_0,y_0), (c_1, x_0,y_0,x_0, y_0))\} \cup$
$\bigcup_{x,y,x_0, y_0 \in \mathbb{N}, x < y} \{((c_1,x,y,x_0,y_0), (c_1,x,y-x,x_0,y_0))\} \cup$
$\bigcup_{x,y,x_0,y_0 \in \mathbb{N}, y < x}  \{((c_1,x,y,x_0,y_0), (c_1, x-y,y,x_0,y_0))\} \cup$
$\bigcup_{x,y,x_0,y_0 \in \mathbb{N}, x = y}\{((c_1,x,y,x_0,y_0), (c_2, x,y,x_0, y_0))\}$.
}

We identify two components of this transition system: $\calS_1$, encoded by the upper self-loop and rightmost arrow, and
$\calS_2$, encoded by the lower self-loop and rightmost arrow. Their state-spaces are both $\{c_1, c_2\} \times \mathbb{N}^{4}$.
Their  transition relations are $\bigcup_{x,y,x_0, y_0 \in \mathbb{N}, x < y} \{((c_1,x,y,x_0,y_0), (c_1,x,y-x,x_0,y_0))\} \cup$
$\bigcup_{x,y,x_0,y_0 \in \mathbb{N}, x = y}\{((c_1,x,y,x_0,y_0), (c_2, x,y,x_0, y_0))\}$ and 
$\bigcup_{x,y,x_0,y_0 \in \mathbb{N}, y < x}  \{((c_1,x,y,x_0,y_0), (c_1, x-y,y,x_0,y_0))\} \cup$
$\bigcup_{x,y,x_0,y_0 \in \mathbb{N}, x = y}\{((c_1,x,y,x_0,y_0), (c_2, x,y,x_0, y_0))\}$,  induced by their respective arrow subsets.
We will show

\smallskip

(1) $\calS \models (c = c_0 \wedge x_0 > 0 \wedge y_0 > 0) \reach (c =c_2 \wedge x = y \wedge x =\mathit{gcd}(x_0,y_0))$.

\smallskip \noindent
which is the functional correctness of the system. Using the soundness of $\Vvdash$ then the $\textsf{[Tra]}$ rule, the latter reduces to (2) $\calS \Vvdash (\false,(c = c_0 \wedge x_0 > 0 \wedge y_0 > 0) \reach (c =c_1 \wedge x = y_0 \wedge y = y_0 \wedge x_0 > 0 \wedge y_0 > 0))$ and 
(3) $\calS \Vvdash (\false,(c = c_1 \wedge x = y_0 \wedge y = y_0 \wedge x_0 > 0 \wedge y_0 > 0) \reach (c =c_2 \wedge  x = y \wedge x =\mathit{gcd}(x_0,y_0)))$. Now, (2) is  discharged by the sequence of rules~\textsf{[Stp]}, \textsf{[Str]} and \textsf{[Trv]}, thus, we focus on (3).
Using several times
\textsf{[Str]} and \textsf{[Spl]}, and also ($x = x_0 \wedge y = y_0) \sqsubseteq  (\mathit{gcd}(x,y)  = \mathit{gcd}(x_0,y_0))$, (3) reduces to proving the   subgoals

(4) :  $\calS \Vvdash (\false,(c_1, \mathit{gcd}(x,y)  = \mathit{gcd}(x_0,y_0) \wedge x_0 > 0 \wedge y_0 > 0 \wedge x < y ) \reach (c = c_2 \wedge x = y \wedge x =\mathit{gcd}(x_0,y_0)))$;

(5) : $\calS \Vvdash (\false,(c_1,  \mathit{gcd}(x,y)  = \mathit{gcd}(x_0,y_0) \wedge x_0 > 0 \wedge y_0 > 0 \wedge x = y) \reach (c = c_2 \wedge x = y \wedge x =\mathit{gcd}(x_0,y_0)))$;

(6) : $\calS \Vvdash (\false,(c_1,  \mathit{gcd}(x,y)  = \mathit{gcd}(x_0,y_0) \wedge x_0 > 0 \wedge y_0 > 0 \wedge y < x) \reach (c = c_2 \wedge x = y \wedge x =\mathit{gcd}(x_0,y_0)))$.

\noindent The subgoal (5) is immediately discharged by applying the sequence of rules~\textsf{[Stp]}, \textsf{[Str]} and \textsf{[Trv]}.

The two other ones we  prove by reducing them, thanks to Theorem~\ref{th:compps3} to the two following subgoals, with
 $\varphi_1 \; \eqbydef \; (c = c_1\wedge \mathit{gcd}(x,y)  = \mathit{gcd}(x_0,y_0) \wedge x_0 > 0 \wedge y_0 > 0 \wedge x < y) \reach (c = c_2 \wedge x = y \wedge x =\mathit{gcd}(x_0,y_0))$ and 
$\varphi_2 \; \eqbydef \; (c = c_1\wedge \mathit{gcd}(x,y)  = \mathit{gcd}(x_0,y_0) \wedge x_0 > 0 \wedge y_0 > 0 \wedge y < x) \reach (c = c_2 \wedge x = y \wedge x =\mathit{gcd}(x_0,y_0))$:

 (7) : $\calS_1, \{(\false, \varphi_2)\} \Vvdash (\false, \varphi_1)$ and (8) : $\calS_2, \{(\false, \varphi_1)\} \Vvdash (\false, \varphi_2)$. We prove (7), the proof of (8) is similar. Using \textsf{[Tra]}, (7) reduces to (9) :
$\calS_1, \{(\false, \varphi_2)\} \Vvdash (\false,  (\phi \wedge x < y) \reach (\phi \wedge y \leq x)))$ and (10) :
$\calS_1, \{(\false, \varphi_2)\} \Vvdash  (\false, (\phi \wedge y \leq x) \reach (c = c_2 \wedge x = y \wedge x =\mathit{gcd}(x_0,y_0)))$ where $\phi \; \eqbydef \; (c= c_1 \wedge \mathit{gcd}(x,y)  = \mathit{gcd}(x_0,y_0) \wedge x_0 > 0 \wedge y_0 > 0)$.

The  subgoal (9) is  proved after simplification by \textsf{[Clr]} using Lemma~\ref{lem:redtoinv3} with $q \; \eqbydef \;  (\phi \wedge x < y) $.

For the subgoal (10), it is first decomposed using \textsf{[Str]} then  \textsf{[Spl]} into (11) : 
$\calS_1, \{(\false, \varphi_2)\} \Vvdash  (\false, (\phi \wedge y = x) \reach (c = c_2 \wedge x = y \wedge x =\mathit{gcd}(x_0,y_0)))$ -  which is easily discharged by  \textsf{[Stp]}, \textsf{[Str]} then  \textsf{[Trv]} -  and

(12) : $\calS_1, \{(\false, \varphi_2)\} \Vvdash  (\false, (\phi \wedge y < x) \reach (c = c_2 \wedge x = y \wedge x =\mathit{gcd}(x_0,y_0)))$. Using \textsf{[Cof]},  (12) becomes

(13) : $\calS_1, \{(\false, \varphi_2), (\false, (\phi \wedge y < x) \reach \psi)\} \Vvdash  (\false, (\phi \wedge y < x) \reach \psi)$ with $\psi \eqbydef (c = c_2 \wedge x = y \wedge x =\mathit{gcd}(x_0,y_0))$.

\noindent We now apply \textsf{[Stp]} followed by \textsf{[Str]} to (13) and get (14) :  $\calS_1, \{(\false, \varphi_2), (\false, (\phi \wedge y < x) \reach \psi)\} \Vvdash  (\true, \phi\reach \psi)$.

After several applications of \textsf{[Str]} and  \textsf{[Spl]}
(14) is reduced to proving the three last following subgoals:

(15) :  $\calS_1, \{(\false, \varphi_2), (\false, (\phi \wedge y < x) \reach \psi)\} \Vvdash  (\true,( \phi \wedge y < x)\reach \psi)$, discharged using \textsf{[Hyp]};

(16) :  $\calS_1, \{(\false, \varphi_2), (\false, (\phi \wedge y = x) \reach \psi)\} \Vvdash  (\true,( \phi \wedge y < x)\reach \psi)$, discharged using \textsf{[Stp]}, \textsf{[Str]}, and \textsf{[Trv]}; 

(17) : $\calS_1, \{(\false, \varphi_2), (\false, (\phi \wedge x < y) \reach \psi)\} \Vvdash  (\true,( \phi \wedge y < x)\reach \psi)$, discharged using \textsf{[Hyp]} by noting that $\varphi_2$ is $( \phi \wedge y < x)\reach \psi$. All the subgoals have been  discharged, and the proof of (1) is complete.

  \end{example}

 \section{Implementations in Isabelle/HOL and Coq}
 We have implemented all the proof systems in Coq and (currently) the first two ones in Isabelle/HOL as well. Our initial goal was to use only Coq, and the reason we also tried Isabelle/HOL (learning it in the process) was that we wanted a ``second opinion'' when faced with difficulties using Coq's coinduction.
 
 The Isabelle/HOL implementation for proof systems $\vdash$ and $\Vdash $is essentially the same as the one described in the paper. The tool automatically generates and proves
 induction and coinduction principles from inductive and coinductive datatypes or predicates.
 Proof commands \texttt{induction} resp. \texttt{coinduction} apply an induction (resp., a coinduction principle) by instantiating the predicate therein via unification with the conclusion, possibly generalised by universally quantifying some variables, (resp., with a conjunction of hypotheses, possibly generalised by existentially quantifying some variables). The overall level of automation is high, which is pleasant to use in practice, the only down side being that users might not understand what is going on. Overall, the proofs in this paper are sketches of the formal Isabelle/HOL proofs, which we did with a lower automation level in order to be able to understand and describe them.

 The Coq implementation for the proof systems $\vdash$ and $\Vdash $ is rather different from the above, because support for coinduction in Coq is also rather different. The standard way to perform a proof by coinduction in Coq is to use the  \texttt{cofix} tactic, which (like the \textsf{[Cof]} rule in our third proof system that emulates it), copies the current goal's conclusion as a new hypothesis, which can only be used after appropriate ``progress'' has been made in the interactive proof. A proof by coinduction in Coq is ultimately
 a well-formed corecursive function, where well-formedness is defined as a  syntactical guardedness condition, which is quite complex  in the theory~\cite{DBLP:conf/types/Gimenez94}, and even more so in the implementation.  
 We have nonetheless managed to prove the soundness and completeness of $\vdash$ using this tactic:
 \textsf{cofix}-style proofs of soundness and completeness for $\vdash$, described in standard mathematical notation,
   are reported in~\cite{rusu:hal-01962912}.
   For $\Vdash$, however, \textsf{cofix} became useless because, for some reason,
   it does not accept to be mixed in a proof by induction. Fortunately, there is a better version, \texttt{pcofix}, part of a
   Coq package called Paco, based on an extenstion of Knaster-Tarski coinduction called \emph{parameterised}
   coinduction~\cite{DBLP:conf/popl/HurNDV13}. Even though the theory is an extension of Knaster-Tarski, anything related to
   fixpoints of functionals is hidden from the user; a set of tactics, including \textsf{pcofix}, leaves the user with the impression that they are using \textsf{cofix} but without its issues.
   
   The soundness proof of $\Vvdash$, only in Coq for now, generally follows the lines shown in this paper.
   It is also completely different from the corresponding proofs for the two other proof systems: it does not use general (co)induction principles, but one well-founded induction principle specific to our problem.

\section{Conclusions and Future Work}
We have presented three proof systems for Reachability Logic on Transition Systems, which use coinduction and induction in different proportions. We have proved their soundness and completeness, and have noted that the more inductive a proof system is, and the more specialised its coinduction style is with respect to our problem domain, the more compositional  the proof system is, but the harder its soundness proof.
Mechanisations of the proof systems in Isabelle/HOL and Coq have also been briefly presented.

In future work we shall make the formal proof of compositionality with respect to transition systems; and prove the third proof system (currently only proved in Coq) in Isabelle/HOL. We are also planning to port Knaster-Tarski coinduction to Coq, and redo the proofs in this paper in that style, in order to obtain Coq proofs closer in spirit to those in the paper and in Isabelle/HOL. A medium-term project is to use the most compositional proof system, among the three proposed ones, for verifying monadic code, a sizeable amount of which is available to us from earlier projects; and, in the longer term, to enrich our proof system with assume-guarantee-style compositional reasoning related to parallel composition.

 \paragraph*{Acknowledgment.}
\hspace*{-0.4cm}  We would like to thank Andrei Popescu for his help with coinduction in Isabelle/HOL.
We acknowledge the support of the CNRS-JSPS Joint Research Project ``FoRmal tools for IoT sEcurity'' (PRC2199), and thank all the participants of this project for fruitful discussions.



\bibliographystyle{eptcs}

\begin{thebibliography}{10}
\providecommand{\bibitemdeclare}[2]{}
\providecommand{\surnamestart}{}
\providecommand{\surnameend}{}
\providecommand{\urlprefix}{Available at }
\providecommand{\url}[1]{\texttt{#1}}
\providecommand{\href}[2]{\texttt{#2}}
\providecommand{\urlalt}[2]{\href{#1}{#2}}
\providecommand{\doi}[1]{doi:\urlalt{http://dx.doi.org/#1}{#1}}
\providecommand{\bibinfo}[2]{#2}

\bibitemdeclare{book}{DBLP:series/txtcs/BertotC04}
\bibitem{DBLP:series/txtcs/BertotC04}
\bibinfo{author}{Yves \surnamestart Bertot\surnameend} \&
  \bibinfo{author}{Pierre \surnamestart Cast{\'{e}}ran\surnameend}
  (\bibinfo{year}{2004}): \emph{\bibinfo{title}{Interactive Theorem Proving and
  Program Development - Coq'Art: The Calculus of Inductive Constructions}}.
\newblock \bibinfo{series}{Texts in Theoretical Computer Science. An {EATCS}
  Series}, \bibinfo{publisher}{Springer}, \doi{10.1007/978-3-662-07964-5}.

\bibitemdeclare{inproceedings}{DBLP:conf/esop/BlanchetteBL0T17}
\bibitem{DBLP:conf/esop/BlanchetteBL0T17}
\bibinfo{author}{Jasmin~Christian \surnamestart Blanchette\surnameend},
  \bibinfo{author}{Aymeric \surnamestart Bouzy\surnameend},
  \bibinfo{author}{Andreas \surnamestart Lochbihler\surnameend},
  \bibinfo{author}{Andrei \surnamestart Popescu\surnameend} \&
  \bibinfo{author}{Dmitriy \surnamestart Traytel\surnameend}
  (\bibinfo{year}{2017}): \emph{\bibinfo{title}{Friends with Benefits -
  Implementing Corecursion in Foundational Proof Assistants}}.
\newblock In: {\sl \bibinfo{booktitle}{{ESOP}}}, {\sl \bibinfo{series}{Lecture
  Notes in Computer Science}} \bibinfo{volume}{10201},
  \bibinfo{publisher}{Springer}, pp. \bibinfo{pages}{111--140},
  \doi{10.1016/0304-3975(91)90043-2}.

\bibitemdeclare{inproceedings}{DBLP:conf/cade/CiobacaL18}
\bibitem{DBLP:conf/cade/CiobacaL18}
\bibinfo{author}{\surnamestart \c{S}tefan Ciob\^ac\u{a}\surnameend} \&
  \bibinfo{author}{Dorel \surnamestart Lucanu\surnameend}
  (\bibinfo{year}{2018}): \emph{\bibinfo{title}{A Coinductive Approach to
  Proving Reachability Properties in Logically Constrained Term Rewriting
  Systems}}.
\newblock In: {\sl \bibinfo{booktitle}{{IJCAR}}}, {\sl \bibinfo{series}{Lecture
  Notes in Computer Science}} \bibinfo{volume}{10900},
  \bibinfo{publisher}{Springer}, pp. \bibinfo{pages}{295--311},
  \doi{10.1016/j.ic.2008.03.026}.

\bibitemdeclare{inproceedings}{DBLP:conf/types/Gimenez94}
\bibitem{DBLP:conf/types/Gimenez94}
\bibinfo{author}{Eduardo \surnamestart Gim{\'{e}}nez\surnameend}
  (\bibinfo{year}{1994}): \emph{\bibinfo{title}{Codifying Guarded Definitions
  with Recursive Schemes}}.
\newblock In: {\sl \bibinfo{booktitle}{{TYPES}}}, {\sl \bibinfo{series}{Lecture
  Notes in Computer Science}} \bibinfo{volume}{996},
  \bibinfo{publisher}{Springer}, pp. \bibinfo{pages}{39--59},
  \doi{10.1007/3-540-60579-7_3}.

\bibitemdeclare{article}{DBLP:journals/cacm/Hoare69}
\bibitem{DBLP:journals/cacm/Hoare69}
\bibinfo{author}{C.~A.~R. \surnamestart Hoare\surnameend}
  (\bibinfo{year}{1969}): \emph{\bibinfo{title}{An Axiomatic Basis for Computer
  Programming}}.
\newblock {\sl \bibinfo{journal}{Commun. {ACM}}}
  \bibinfo{volume}{12}(\bibinfo{number}{10}), pp. \bibinfo{pages}{576--580},
  \doi{10.1145/363235.363259}.

\bibitemdeclare{inproceedings}{DBLP:conf/popl/HurNDV13}
\bibitem{DBLP:conf/popl/HurNDV13}
\bibinfo{author}{Chung{-}Kil \surnamestart Hur\surnameend},
  \bibinfo{author}{Georg \surnamestart Neis\surnameend}, \bibinfo{author}{Derek
  \surnamestart Dreyer\surnameend} \& \bibinfo{author}{Viktor \surnamestart
  Vafeiadis\surnameend} (\bibinfo{year}{2013}): \emph{\bibinfo{title}{The power
  of parameterization in coinductive proof}}.
\newblock In: {\sl \bibinfo{booktitle}{{POPL}}}, \bibinfo{publisher}{{ACM}},
  pp. \bibinfo{pages}{193--206}.

\bibitemdeclare{article}{DBLP:journals/jsc/LucanuRA17}
\bibitem{DBLP:journals/jsc/LucanuRA17}
\bibinfo{author}{Dorel \surnamestart Lucanu\surnameend}, \bibinfo{author}{Vlad
  \surnamestart Rusu\surnameend} \& \bibinfo{author}{Andrei \surnamestart
  Arusoaie\surnameend} (\bibinfo{year}{2017}): \emph{\bibinfo{title}{A generic
  framework for symbolic execution: {A} coinductive approach}}.
\newblock {\sl \bibinfo{journal}{J. Symb. Comput.}} \bibinfo{volume}{80}, pp.
  \bibinfo{pages}{125--163}, \doi{10.1016/j.jsc.2016.07.012}.

\bibitemdeclare{inproceedings}{DBLP:conf/birthday/LucanuRAN15}
\bibitem{DBLP:conf/birthday/LucanuRAN15}
\bibinfo{author}{Dorel \surnamestart Lucanu\surnameend}, \bibinfo{author}{Vlad
  \surnamestart Rusu\surnameend}, \bibinfo{author}{Andrei \surnamestart
  Arusoaie\surnameend} \& \bibinfo{author}{David \surnamestart
  Nowak\surnameend} (\bibinfo{year}{2015}): \emph{\bibinfo{title}{Verifying
  Reachability-Logic Properties on Rewriting-Logic Specifications}}.
\newblock In: {\sl \bibinfo{booktitle}{Logic, Rewriting, and Concurrency}},
  {\sl \bibinfo{series}{Lecture Notes in Computer Science}}
  \bibinfo{volume}{9200}, \bibinfo{publisher}{Springer}, pp.
  \bibinfo{pages}{451--474}, \doi{10.1007/978-3-319-02654-1_16}.

\bibitemdeclare{inproceedings}{DBLP:conf/esop/MoorePR18}
\bibitem{DBLP:conf/esop/MoorePR18}
\bibinfo{author}{Brandon~M. \surnamestart Moore\surnameend},
  \bibinfo{author}{Lucas \surnamestart Pe{\~{n}}a\surnameend} \&
  \bibinfo{author}{Grigore \surnamestart Rosu\surnameend}
  (\bibinfo{year}{2018}): \emph{\bibinfo{title}{Program Verification by
  Coinduction}}.
\newblock In: {\sl \bibinfo{booktitle}{{ESOP}}}, {\sl \bibinfo{series}{Lecture
  Notes in Computer Science}} \bibinfo{volume}{10801},
  \bibinfo{publisher}{Springer}, pp. \bibinfo{pages}{589--618},
  \doi{10.1145/2480359.2429093}.

\bibitemdeclare{book}{DBLP:books/sp/NipkowPW02}
\bibitem{DBLP:books/sp/NipkowPW02}
\bibinfo{author}{Tobias \surnamestart Nipkow\surnameend},
  \bibinfo{author}{Lawrence~C. \surnamestart Paulson\surnameend} \&
  \bibinfo{author}{Markus \surnamestart Wenzel\surnameend}
  (\bibinfo{year}{2002}): \emph{\bibinfo{title}{Isabelle/HOL - {A} Proof
  Assistant for Higher-Order Logic}}.
\newblock {\sl \bibinfo{series}{Lecture Notes in Computer Science}}
  \bibinfo{volume}{2283}, \bibinfo{publisher}{Springer},
  \doi{10.1007/3-540-45949-9_6}.

\bibitemdeclare{article}{DBLP:journals/cacm/OHearn19}
\bibitem{DBLP:journals/cacm/OHearn19}
\bibinfo{author}{Peter~W. \surnamestart O'Hearn\surnameend}
  (\bibinfo{year}{2019}): \emph{\bibinfo{title}{Separation logic}}.
\newblock {\sl \bibinfo{journal}{Commun. {ACM}}}
  \bibinfo{volume}{62}(\bibinfo{number}{2}), pp. \bibinfo{pages}{86--95},
  \doi{10.1145/3211968}.

\bibitemdeclare{book}{DBLP:books/cu/RoeverBH2001}
\bibitem{DBLP:books/cu/RoeverBH2001}
\bibinfo{author}{Willem~P. \surnamestart de~Roever\surnameend},
  \bibinfo{author}{Frank~S. \surnamestart de~Boer\surnameend},
  \bibinfo{author}{Ulrich \surnamestart Hannemann\surnameend},
  \bibinfo{author}{Jozef \surnamestart Hooman\surnameend},
  \bibinfo{author}{Yassine \surnamestart Lakhnech\surnameend},
  \bibinfo{author}{Mannes \surnamestart Poel\surnameend} \&
  \bibinfo{author}{Job \surnamestart Zwiers\surnameend} (\bibinfo{year}{2001}):
  \emph{\bibinfo{title}{Concurrency Verification: Introduction to Compositional
  and Noncompositional Methods}}.
\newblock {\sl \bibinfo{series}{Cambridge Tracts in Theoretical Computer
  Science}}~\bibinfo{volume}{54}, \bibinfo{publisher}{Cambridge University
  Press}.

\bibitemdeclare{inproceedings}{DBLP:conf/lics/RosuSCM13}
\bibitem{DBLP:conf/lics/RosuSCM13}
\bibinfo{author}{Grigore \surnamestart Rosu\surnameend},
  \bibinfo{author}{Andrei \surnamestart Stefanescu\surnameend},
  \bibinfo{author}{\surnamestart \c{S}tefan Ciob\^ac\u{a}\surnameend} \&
  \bibinfo{author}{Brandon~M. \surnamestart Moore\surnameend}
  (\bibinfo{year}{2013}): \emph{\bibinfo{title}{One-Path Reachability Logic}}.
\newblock In: {\sl \bibinfo{booktitle}{{LICS}}}, \bibinfo{publisher}{{IEEE}
  Computer Society}, pp. \bibinfo{pages}{358--367}.

\bibitemdeclare{inproceedings}{DBLP:conf/tase/RusuGH18}
\bibitem{DBLP:conf/tase/RusuGH18}
\bibinfo{author}{Vlad \surnamestart Rusu\surnameend}, \bibinfo{author}{Gilles
  \surnamestart Grimaud\surnameend} \& \bibinfo{author}{Micha{\"{e}}l
  \surnamestart Hauspie\surnameend} (\bibinfo{year}{2018}):
  \emph{\bibinfo{title}{Proving Partial-Correctness and Invariance Properties
  of Transition-System Models}}.
\newblock In: {\sl \bibinfo{booktitle}{{TASE}}}, \bibinfo{publisher}{{IEEE}
  Computer Society}, pp. \bibinfo{pages}{60--67}.

\bibitemdeclare{unpublished}{rusu:hal-01962912}
\bibitem{rusu:hal-01962912}
\bibinfo{author}{Vlad \surnamestart Rusu\surnameend}, \bibinfo{author}{Gilles
  \surnamestart Grimaud\surnameend} \& \bibinfo{author}{Micha{\"e}l
  \surnamestart Hauspie\surnameend} (\bibinfo{year}{2019}):
  \emph{\bibinfo{title}{{Proving Partial-Correctness and Invariance Properties
  of Transition-System Models}}}.
\newblock \urlprefix\url{https://hal.inria.fr/hal-01962912}.

\bibitemdeclare{book}{sangiorgi2011}
\bibitem{sangiorgi2011}
\bibinfo{author}{Davide \surnamestart Sangiorgi\surnameend}
  (\bibinfo{year}{2011}): \emph{\bibinfo{title}{Introduction to Bisimulation
  and Coinduction}}.
\newblock \bibinfo{publisher}{Cambridge University Press},
  \bibinfo{address}{New York, NY, USA}, \doi{10.1017/CBO9780511777110}.

\bibitemdeclare{inproceedings}{DBLP:conf/lopstr/SkeirikSM17}
\bibitem{DBLP:conf/lopstr/SkeirikSM17}
\bibinfo{author}{Stephen \surnamestart Skeirik\surnameend},
  \bibinfo{author}{Andrei \surnamestart Stefanescu\surnameend} \&
  \bibinfo{author}{Jos{\'{e}} \surnamestart Meseguer\surnameend}
  (\bibinfo{year}{2017}): \emph{\bibinfo{title}{A Constructor-Based
  Reachability Logic for Rewrite Theories}}.
\newblock In: {\sl \bibinfo{booktitle}{{LOPSTR}}}, {\sl
  \bibinfo{series}{Lecture Notes in Computer Science}} \bibinfo{volume}{10855},
  \bibinfo{publisher}{Springer}, pp. \bibinfo{pages}{201--217},
  \doi{10.1007/978-3-319-08918-8_29}.

\bibitemdeclare{article}{DBLP:journals/lmcs/StefanescuCMMSR19}
\bibitem{DBLP:journals/lmcs/StefanescuCMMSR19}
\bibinfo{author}{Andrei \surnamestart Stefanescu\surnameend},
  \bibinfo{author}{\surnamestart \c{S}tefan Ciob\^ac\u{a}\surnameend},
  \bibinfo{author}{Radu \surnamestart Mereuta\surnameend},
  \bibinfo{author}{Brandon~M. \surnamestart Moore\surnameend},
  \bibinfo{author}{Traian{-}Florin \surnamestart Serbanuta\surnameend} \&
  \bibinfo{author}{Grigore \surnamestart Rosu\surnameend}
  (\bibinfo{year}{2019}): \emph{\bibinfo{title}{All-Path Reachability Logic}}.
\newblock {\sl \bibinfo{journal}{Logical Methods in Computer Science}}
  \bibinfo{volume}{15}(\bibinfo{number}{2}).

\bibitemdeclare{inproceedings}{DBLP:conf/oopsla/StefanescuPYLR16}
\bibitem{DBLP:conf/oopsla/StefanescuPYLR16}
\bibinfo{author}{Andrei \surnamestart Stefanescu\surnameend},
  \bibinfo{author}{Daejun \surnamestart Park\surnameend},
  \bibinfo{author}{Shijiao \surnamestart Yuwen\surnameend},
  \bibinfo{author}{Yilong \surnamestart Li\surnameend} \&
  \bibinfo{author}{Grigore \surnamestart Rosu\surnameend}
  (\bibinfo{year}{2016}): \emph{\bibinfo{title}{Semantics-based program
  verifiers for all languages}}.
\newblock In: {\sl \bibinfo{booktitle}{{OOPSLA}}}, \bibinfo{publisher}{{ACM}},
  pp. \bibinfo{pages}{74--91}, \doi{10.1145/2983990.2984027}.

\end{thebibliography}

\end{document}